\documentclass[11pt,twoside]{article} 
\usepackage{euler,acta-info}
\usepackage[latin2]{inputenc}
\title{%
On Erd\H os-Gallai and Havel-Hakimi algorithms
}

\newcommand{\vol}{\textbf{3,} 2 (2011) 230--268}   
\newcommand{\amss}{\amssubj{%
05C85, 68R10
}}   
\newcommand{\CRs}{\CRsubj{%
G.2.2. [Graph Theory]: Subtopic - Network problems.
}}   
\newcommand{\keyw}{\keywords{%
simple graphs, prescribed degree sequences, Erd\H os-Gallai theorem, Havel-Hakimi theorem, graphical sequences
}}

\newcommand{\ideze}{\setbox0=\hbox{\lower1.38ex\hbox{''}}\dp0=0pt\box0}

\begin{document}
\maketitle

\twoauthors{%
\href{http://compalg.inf.elte.hu/tanszek/tony/oktato.php?oktato=tony&angolul=1}{Antal Iv\'anyi }}
{%
\href{http://compalg.inf.elte.hu/tanszek/index.php?angolul=1}{Department of Computer Algebra}, 
\href{http://www.elte.hu/en}{E\"otv\"os Lor\'and University}, Hungary}
{%
 \href{mailto:tony@compalg.inf.elte.hu}{tony@inf.elte.hu}
}{%
\href{http://people.inf.elte.hu/lulsaai}{Lor\'and Lucz}
}{%
 \href{http://compalg.inf.elte.hu/tanszek/index.php?angolul=1}{Department of Computer Algebra}, 
 \href{http://www.elte.hu/en}{E\"otv\"os Lor\'and University}, Hungary}
{%
 \href{mailto:lorand.lucz@gmail.com}{lorand.lucz@gmail.com}}
 
\twoauthors{%
\href{http://www.math.elte.hu/~mori}{Tam\'as F. M\'ori}}
{%
\href{http://www.cs.elte.hu/probability/common/index.a.html}{Department of Probability Theory and Statistics},
\href{http://www.elte.hu/en}{Eötvös Loránd University}, Hungary}  
{%
\href{mailto:moritamas@ludens.elte.hu}{moritamas@ludens.elte.hu}
}{%
 \href{http://people.inf.elte.hu/sopsaai}{P\'eter S\'ot\'er}
}{%
 \href{http://compalg.inf.elte.hu/tanszek/index.php?angolul=1}{Department of Computer Algebra}, 
 \href{http://www.elte.hu/en}{Eötvös Loránd University}, Hungary}
{%
 \href{mapoleon@freemail.hu}{mapoleon@freemail.hu}}


\short{%
A. Iv\'anyi, L. Lucz, T. F. M\'ori, P. S\'ot\'er
}{%
On Erd\H os-Gallai and Havel-Hakimi algorithms}

\begin{abstract}
Havel in 1955 \cite{Havel1955}, Erd\H os and Gallai in 1960 \cite{ErdosG1960}, 
Hakimi in 1962 \cite{Hakimi1962}, Ruskey, Cohen, Eades and Scott in 1994 \cite{RuskeyCES1994}, Barnes and Savage in 1997  
\cite{BarnesS1997}, Kohnert in 2004 \cite{Kohnert2004}, Tripathi, Venugopalan and West in 2010 
\cite{TripathiVW2010} proposed a method to decide, whether a sequence    
of nonnegative integers can be the degree sequence of a simple graph. 
The running time of their algorithms is $\Omega(n^2)$ in worst case. In this paper we propose a new algorithm called EGL 
(Erd\H os-Gallai Linear algorithm), whose worst running time is $\Theta(n).$ As an application of this quick algorithm we 
computed the number of the different degree sequences of simple graphs for 
$24, \ \ldots, \ 29$ vertices (see \cite{Sloane2011A004251}).
\end{abstract}

\normalsize
\section{Introduction\label{sec-Intro}}
In the practice an often appearing problem is the ranking of different objects as hardware or software products, cars, 
economical decisions, persons etc. A typical method of the ranking is the pairwise comparison of the objects, assignment 
of points to the objects and sorting the objects according to the sums of the numbers of the received points.

For example Landau \cite{Landau1953} references to biological, Hakimi \cite{Hakimi1962} to chemical, Kim et al. \cite{KimTME2009},  
Newman and Barabási \cite{NewmanB2006} to net-centric, Boz\'oki, F\"ul\"op, Poesz, K\'eri, R\'onyai and Temesi to economical 
\cite{AnholzerBBK2011,BozokiFP2011,BozokiFR2010,Keri2011,Temesi2011}, Liljeros et al. \cite{LiljerosEASA2001} to human applications, 
while Iv\'anyi, Khan, Lucz, Pirzada, S\'ot\'er and Zhou \cite{Ivanyi2009,Ivanyi2010,IvanyiP2011,PirzadaIK2011,PirzadaZI2010} 
write on applications in sports.  

From several popular possibilities we follow the terminology and notations used by Paul Erd\H os and Tibor Gallai \cite{ErdosG1960}. 

Depending from the rules of the allocation of the points there are many problems. In this paper we deal only with the case when 
the comparisons have two possible results: either both objects get one point, or both objects get zero points.
In this case the results of the comparisons can be represented using simple graphs and the number of points gathered by the given 
objects are the degrees of the corresponding vertices. The decreasing sequence of the degrees is denoted by $b = (b_1,\ldots,b_n)$. 

From the popular problems we investigate first of all the question, how can we quickly decide, whether for given $b$ does exist 
there a simple graph $G$ whose degree sequence is $b$. In connection with this problem we remark that the main motivation 
for studying of this problem is the question: what is the complexity of deciding whether a sequence is the score sequence 
of some football tournament \cite{Frank2011,Ivanyi2011Egres,IvanyiL2011Comb,IvanyiL2011CEJOR,KernP2001,KernP2004,LuczISP2012MaCS}. 

As a side effect we extended the popular data base \textit{On-line Encyclopedia of Integer Sequences} \cite{Sloane2011}
with the continuation of contained sequences.

In connection with the similar problems we remark, that in the last years a lot of papers and chapters were published 
on the undirected graphs (for example \cite{BeregI2008,Bodei2010,BrualdiK2009,CooperL2011,delGenioKTB2010,
HellK2009,IvanyiLMS2011AML,
KayibiKPI2011,Meierling2009,RodsethST2009,TripathiT2008,TripathiVW2010,Weisstein2011DS,Weisstein2011GS}) and also on directed graphs 
(for example \cite{BeasleyBR2009,BozokiFR2010,BuschCJ2010,ErdosMT2010,Frank2011,Ivanyi2009,Ivanyi2010,Ivanyi2011Kyoto,IvanyiP2011,
KimTME2009,Knuth2011,LaMar2010,Miklos2009,MiklosES2011,Pirzada2011,PirzadaIK2011,PirzadaI2012,PirzadaNSI2010}).

The majority of the investigated algorithms is sequential, but there are parallel results too 
\cite{ArikatiM1996,DeAgostinoP1990,DessmarkLG1994,IvanyiL2011CEJOR,NarayanaB1964,PecsySz2000,Soroker1990}. 

Let $l, \ u$ and $m$ integers $(m \geq 1$ and $u \geq l).$ A sequence of integer numbers  $b = (b_1,\ldots,b_m)$ 
is called $(l,u,m)$\textit{-bounded,} if $l \leq b_i \leq u$ for $i = 1,\ldots, m.$ A $(l,u,m)$\textit{-bounded} sequence $b$ is 
called $(l,u,m)$\textit{-regular,} if $b_m \geq b_{m-1} \geq \cdots \geq b_1.$ An $(l,u,m)$-regular sequence is called 
$(l,u,m)$\textit{-even}, if the sum of its elements is even. A $(0,n-1,n)$\textit{-regular} sequence $b$ 
is called $n$\textit{-graphical,} if there exists a simple graph $G$ whose degree sequence is $b.$ If $l = 0, \ u = n - 1$ and 
$m = n,$ then we use the terms $n$-bounded, $n$-regular, $n$-even, and $n$-graphical (or simply bounded, regular, even, graphical). 

In the following we deal first of all with regular sequences. In our definitions the bounds appear to save the testing algorithms 
from the checking of such sequences, which are obviously not graphical, therefore these bounds do not mean 
the restriction of the generality.  

The paper consists of nine parts. After the introductory Section \ref{sec-Intro} in Section \ref{sec-classical} 
we describe the classical algorithms of the testing and reconstruction of degree sequences of simple graphs.  
Section \ref{sec-testing} introduces several linear testing algorithms, then Section \ref{sec-approx} 
summarizes some properties of the approximate algorithms. Section \ref{sec-newprecise} contains the description of new precise 
algorithms and in Section \ref{sec-classrun} the running times of the classical testing 
algorithms are presented. Section \ref{sec-enum} contains enumerative results, in Section \ref{sec-gamma} we report 
on the application of the new algorithms for the computation of the number of score sequences of simple graphs.  
Finally Section \ref{sec-summary} contains the summary of the results.

Our paper \cite{IvanyiLMS2011AML} written in Hungarian contains further algorithms and simulation results.  
\cite{IvanyiL2011Comb} contains a short summary on the linear Erd\H os-Gallai algorithm while in \cite{IvanyiL2011CEJOR} 
the details of the parallel implementation of enumerating Erd\H os-Gallai algorithm are presented.

\section{Classical precise algorithms\label{sec-classical}}
For a given $n$-regular sequence $b = (b_1,\ldots,b_n)$ the first $i$ elements of the sequence we call 
\textit{the head} of the sequence belonging to the index $i,$ while the last $n - i$ elements of the sequence we call 
\textit{the tail} of the sequence belonging to the index $i.$ 

\subsection{Havel-Hakimi algorithm\label{subsec-HH}}
The first algorithm for the solution of the testing problem was proposed by Vaclav Havel Czech mathematician 
\cite{Havel1955,Lovasz2007}. In 1962 Louis Hakimi \cite{Hakimi1962} published independently the same result, therefore 
the theorem is called today usually as  \textit{Havel-Hakimi theorem}, and the method of reconstruction 
is called \textit{Havel-Hakimi algorithm}. 

\begin{theorem} \emph{(Hakimi \cite{Hakimi1962}, Havel \cite{Havel1955})}. If $n \geq 3$, then the $n$-regular sequence 
$b = (b_1,\ldots,b_n)$ is $n$-graphical if and only if the sequence $b' = (b_2 - 1, b_3 - 1, \ldots,  b_{b_1} - 1, b_{b_1 + 1} - 1, 
b_{b_1 + 2}, \ldots, b_n)$ is $(n - 1)$-graphical.
\end{theorem}

\begin{proof} See \cite{Bodei2010,Hakimi1962,Havel1955,IvanyiLMS2011AML}.
\end{proof}

If we write a recursive program based on this theorem, then according to the RAM model of computation its running time 
will be in worst case  $\Omega(n^2)$, since the algorithm decreases the degrees by one, 
and e.g. if $b = ((n - 1)^n)$, then the sum of the elements of $b$ equals to $\Theta(n^2)$. It is worth 
to remark that the proof of the theorem is constructive, and the algorithm based on the proof not only tests 
the input in quadratic time, but also construct a corresponding simple graph (of course, only if it there exists).

It is worth to remark that the algorithm was extended to directed graphs in which any pair of the vertices is connected 
with at least $a \geq 0$ and at most $b \geq a$ edges \cite{Ivanyi2009,Ivanyi2010}. The special case $a = b = 1$ 
was reproved in \cite{ErdosMT2010}.

In 1965 Hakimi \cite{Hakimi1965} gave a necessary and sufficient condition for \textit{two sequences} 
$a  = (a_1, \ldots,a_n)$ and $b = (b_1,\ldots,b_n)$ to be the in-degree sequences and out-degree sequence 
of a directed multigraph without loops.

\subsection{Erd\H os-Gallai algorithm\label{subsec-EG}} 
In chronological order the next result is the necessary and sufficient theorem published by Paul Erd\H os 
and Tibor Gallai \cite{ErdosG1960}.

For an $n$-regular sequence $b = (b_1, \ldots, b_n)$ let $H_i = b_1 + \cdots + b_i$.  
For given $i$ the elements $b_1, \ \ldots, \ b_i$ are called  \textit{the head} of $b,$ belonging to $i,$ while   
the elements $b_{i+1}, \ \ldots, \ b_n$ are called \textit{the tail} of $b$ belonging to $i.$ 

When we investigate the realizability of a sequence, a natural observation is that the degree requirement $H_i$ of a head 
is covered partially with inner and partially with outer degrees (with edges among the vertices of the head, resp. with 
edges, connecting a vertex of the head and a vertex of the tail). This observation is formalized 
by the following Erd\H os-Gallai theorem.

\begin{theorem} \emph{(Erd\H os, Gallai \cite{ErdosG1960})} Let \label{theorem-EG} $n \geq 3.$ The $n$-regular sequence 
$b = (b_1, \ldots,b_n)$ is $n$-graphical if and only if 
\begin{equation}
\sum _{i=1}^n b_i \quad \mbox{even}\label{eq-EGparity}
\end{equation}
and
\begin{equation}
\sum_{i=1}^j b_i \leq j(j - 1) + \sum _{k=j+1}^n \min(j,b_k) \quad (j = 1,\ldots,n-1).\label{eq-EGbinom}
\end{equation}
\end{theorem}

\begin{proof} See \cite{Bodei2010,Choudum1986,ErdosG1960,SierksmaH1991,TripathiVW2010}.
\end{proof}

\begin{figure}[!t]
\begin{footnotesize}
\begin{center}
\begin{tabular}{|r|r|r|r|}  \hline \hline
$n$  &    $R(n)$  &$E(n)$& $E(n)/R(n)$       \\ \hline \hline
$1$  &     $1$       &$1$          &       $1.0000000000000$            \\ \hline
$2$  &     $3$       &$2$          &      $0.6666666666667$  \\ \hline
$3$  &     $10$      &     $6$     &      $0.6000000000000$       \\ \hline
$4$  &     $35$      &     $19$    &  $0.5428571428571$\\ \hline
$5$  &     $126$     &     $66$    &    $0.5238095238095$           \\ \hline
$6$  &     $462$     &     $236$   &    $0.5108225108225$           \\ \hline 
$7$  &     $1716$    &     $868$   &    $0.5058275058275$         \\ \hline 
$8$  &     $6435$    &    $3235$   &    $0.5027195027195$     \\ \hline 
$9$  &     $24310$   &   $12190$   &    $0.5014397367339$      \\ \hline 
$10$  &    $92378$   &   $46252$   &    $0.5006819805581$   \\ \hline 
$11$  &   $352716$   &  $176484$   &    $0.5003572279114$    \\ \hline 
$12$  & $1352078$    &  $676270$   &    $0.5001708481315$     \\ \hline 
$13$  & $5200300$    & $2600612$   &    $0.5000888410284$      \\ \hline 
$14$  &$20058300$    &$10030008$   &    $0.5000427753100$      \\ \hline 
$15$  &$77558760$    &$38781096$   &    $0.5000221251603$     \\ \hline 
$16$  &$300540195$   &$150273315$  &    $0.5000107057227$     \\ \hline 
$17$  &$1166803110$  &$583407990$  &    $0.5000055150693$      \\ \hline 
$18$  &$4537567650$  &$2268795980$ &    $0.5000026787479$     \\ \hline 
$19$  &$17672631900$ &$8836340260$ &    $0.5000013755733$      \\ \hline 
$20$  &$68923264410$ &$34461678394$&    $0.5000006701511$       \\ \hline 
$21$  &$269128937220$&$134564560988$&   $0.5000003432481$      \\ \hline 
$22$&$1052049481860$&$526024917288$&    $0.5000001676328$      \\ \hline 
$23$&$4116715363800$&$2058358034616$&   $0.5000000856790$    \\ \hline 
$24$&$16123801841550$  &$8061901596814$   &$0.5000000419280$        \\ \hline 
$25$&$63205303218876$  &$31602652961516$  &$0.5000000213918$  \\ \hline 
$26$&$247959266474052$ &$123979635837176$  &    $0.5000000104862$     \\ \hline 
$27$&$973469712824056$ &$486734861612328$  &    $0.5000000053420$      \\ \hline 
$28$&3824345300380220&1912172660219260 &    $0.5000000026224$     \\ \hline 
$29$&15033633249770520&7516816644943560 &    $0.5000000013342$     \\ \hline 
$30$&59132290782430712 &29566145429994736&    $0.5000000006558$  \\ \hline 
$31$&$232714176627630544$             &$116357088391374032$            &   $0.5000000003333$      \\ \hline 
$32$ &$916312070471295267$&  $458156035385917731$&   $0.5000000001640$      \\ \hline 
$33$ &$3609714217008132870$&$1804857108804606630$&   $0.5000000000833$    \\ \hline 
$34$&$14226520737620288370$&$7113260369393545740$          &   $0.5000000000410$           \\ \hline 
$35$&$56093138908331422716$&$28046569455332514468$ &   $0.5000000000208$     \\ \hline  
$36$  &$221256270138418389602$&$ 110628135071477978626$&   $0.5000000000103$         \\ \hline 
$37$ &$873065282167813104916$&  $436532641088444120108$&   $0.5000000000052$        \\ \hline 
$38$ &$3446310324346630677300$&$1723155162182151654600$&   $0.5000000000026$      \\ \hline \hline
\end{tabular}
\end{center}
\end{footnotesize}
\caption{Number of regular and even sequences, and the ratio of these numbers\label{fig-regevenratio}}
\end{figure}

Although this theorem does not solve the problem of reconstruction of graphical sequences, the systematic application of 
(\ref{eq-EGbinom}) requires in worst case (for example when the input sequence is graphical) $\Theta(n^2)$ time. 

Recently Tripathi and Vijay \cite{TripathiVW2010} published a constructive proof of Erd\H os-Gallai theorem and proved that 
their construction requires $O(n^3)$ time.

Figure \ref{fig-regevenratio} shows the number of $n$-regular $(R(n)$ and $n$-even $(E(n)$ sequences and their ratio 
$(E(n)/R(n)$ for $n = 1, \ \ldots, \ 38.$ According to \eqref{eq-ERlim} the sequence of these ratios tends to $\frac{1}{2}$ as 
$n$ tends to $\infty$. According to Figure \ref{fig-regevenratio} the convergence is quick: e.g. $E(20)/R(20) = 0.5000006701511$.

The pseudocode of \textsc{Erd\H os-Gallai} see in \cite{IvanyiLMS2011AML}.

\section{Testing algorithms\label{sec-testing}}
We are interested in the investigation of football sequences, where often appears the necessity of the testing of 
degree sequences of simple graphs. 

A possible way to decrease the expected testing time is to use quick (linear) filtering algorithms which can state with a 
high probability, that the given input is not graphical, and so we need the slow precise algorithms only in the remaining cases. 

Now we describe a \textit{parity checking}, then a \textit{binomial,} and finally a 
\textit{headsplitting} filtering algorithm.

\subsection{Parity test\label{subsec-parity}}
Our first test is based on the first necessary condition of Erd\H os-Gallai theorem. This test is very effective,  
since according to Figure \ref{fig-regevenratio} and Corollary \ref{cor-ER}  about the half of the regular 
sequences is odd, and our test establishes in linear time, that these sequences are not graphical.

The following simple algorithm is based on (\ref{eq-EGparity}).

\textit{Input}. $n $: number of the vertices $(n \geq 1)$; \\
$b = (b_1, \ldots, b_n)$: an $n$-regular sequence.

\textit{Output}. $L$: logical variable ($L = \textsc{False}$ shows, that $b$ is not graphical, while the meaning of the value   
$L = \textsc{True}$ is, that the test \textit{could not decide,} whether $b$ is graphical or not).

\textit{Working variable}. $i$: cycle variable; \\
$H = (H_1, \ldots, H_n)$: $H_i$ is the sum of the first $i$ elements of $b$.

\bigskip
\noindent \textsc{Parity-Test}$(n,b,L)$
\vspace{-2mm}
\begin{tabbing}%
199 \= xxx\=xxx\=xxx\=xxx\=xxx\=xxx\=xxx\=xxx \+ \kill 
\hspace{-7mm}01 $H_1 = b_1$ \\
\hspace{-7mm}02 \textbf{for} \= $i = 2$ \textbf{to} $n$ \\
\hspace{-7mm}03              \> $H_i = H_{i - 1} + b_i$ \\
\hspace{-7mm}04 \textbf{if} \= $H_n$ odd \\
\hspace{-7mm}05             \> $L = \textsc{False}$ \\
\hspace{-7mm}06             \> \textbf{return} $L$ \\
\hspace{-7mm}07 $L = \textsc{True}$ \\
\hspace{-7mm}08 \textbf{return} $L$
\end{tabbing}

The running time of this algorithm is $\Theta(n)$ in all cases. Figure \ref{fig-regevenratio} 
characterizes the efficiency of \textsc{Parity-test}.

\eqref{eq-EGparity} is only a necessary condition, 
therefore \textsc{Parity-Test} is only an approximate (filtering) algorithm. 

\subsection{Binomial test\label{subsec-binom}}
Our second test is based on the second necessary condition of Erd\H os-Gallai theorem. It received the given name since 
we estimate the number of the inner edges of the head of $b$ using a binomial coefficient. 
Let $T_i = b_{i+1} + \cdots + b_n  \ (i = 1, \ldots, n)$.

\begin{lemma} If\label{lemma-binom} $n \geq 1$ and $b$ is an $n$-graphical sequence, then 
\begin{equation} 
H_i  \leq i(i - 1) + T_i  \quad (i = 1, \ldots, n - 1).\label{eq-binom}
\end{equation}
\end{lemma}

\begin{proof} The left side of (\ref{eq-binom}) represents the degree requirement of the head of $b$. 
On the right side of (\ref{eq-binom}) $i(i - 1)$ is an upper bound of the inner degree capacity of the head, 
while $T_i$ is an upper bound of the degree capacity of the tail, belonging to the index $i.$
\end{proof}

The following program is based on Lemma \ref{lemma-binom}.

\textit{Input}. $n $: number of the vertices $(n \geq 1)$; \\
$b = (b_1, \ldots, b_n)$: an $n$-regular even sequence; \\
$H = (H_1, \ldots, H_n)$: $H_i$ the sum of the first $i$ elements of $b$; \\
$H_0$: auxiliary variable, helping to compute the elements of $H$.

\textit{Output}. $L$: logical variable ($L = \textsc{False}$ signals, that $b$ is surely not graphical, while  
$L = \textsc{True}$ shows, that the test \textit{could not decide,} whether $b$ is graphical).

\textit{Working variables.} $i$: cycle variable; \\
$T = (T_1, \ldots, T_n)$: $T_i$ the sum of the last $n - i$ elements of $b$; \\
$T_0$: auxiliary variable, helping to compute the elements of $T$.

\bigskip
\noindent \textsc{Binomial-Test}$(n,b,H,L)$
\vspace{-2mm}
\begin{tabbing}%
199 \= xxx\=xxx\=xxx\=xxx\=xxx\=xxx\=xxx\=xxx \+ \kill 
\hspace{-7mm}01 $T_0 = 0$ \\
\hspace{-7mm}02 \textbf{for} \= $i = 1$ \textbf{to} $n - 1$ \\
\hspace{-7mm}03              \> $T_i = H_n - H_i$ \\
\hspace{-7mm}04              \> \textbf{if} \= $H_i > i(i - 1) + T_i$ \\
\hspace{-7mm}05              \>             \> $L = \textsc{False}$ \\
\hspace{-7mm}06              \>             \> \textbf{return} $L$ \\          
\hspace{-7mm}07 $L = \textsc{True}$ \\
\hspace{-7mm}08 \textbf{return} $L$
\end{tabbing}

The running time of this algorithm is $\Theta(n)$ in worst case, while in best case is only $\Theta(1)$. 

According to our simulation experiments \textsc{Binomial-Test} is an effective filtering test (see Figure \ref{fig-numberofseq} and Figure \ref{fig-apprrelativ}).    

\subsection{Splitting of the head\label{subsec-headsplitting}}
We can get a better estimation of the inner capacity of the head, than the binomial coefficient gives in (\ref{eq-binom}), if we 
split the head into two parts. Let $\lfloor i/2 \rfloor = h_i$, $p$ the number of positive elements of $b$.
Then the sequence $(b_1,\ldots,b_{h_i})$ is called \textit{the beginning of the head} belonging to index $i$ and the 
sequence $(b_{h_i+1},\ldots,b_i)$ \textit{the end of the head} belonging to index $i$.

\begin{lemma} If\label{lemma-headsplitting} $n \geq 1$ and $b$ is an $n$-graphical sequence, then 
\begin{align}
H_i & \leq \min(\min(H_{h_i},T_n - T_i,h_i(n - i))\notag \\
    & + \min(H_i - H_{h_i},T_n - T_i,(i - h_i)(n - i)),T_i) \notag \\
    & + \min(h_i(i - h_i) + \binom{h_i}{2} + \binom{i - h_i}{2}  \quad (i = 1, \ldots, n),\label{eq-HSa}
\end{align}
further
\begin{equation}
\min(H_{h_i},T_n - T_i,h_i(n - i)) + \min(H_i - H_{h_i},T_n - T_i,(i - h_i)(n - i)) \leq T_i.\label{eq-HSb}
\end{equation}
\end{lemma}

\begin{proof} Let $G$ be a simple graph whose degree sequence is $b$. Then we divide the set of the edges of the head 
belonging to index $i$ into five subsets: $(S_{i,1})$ contains the edges between the beginning of the head and the tail, 
$(S_{i,2})$ the edges between the end of the head and the tail, $S_{i,3}$ the edges between the parts of the head, 
$S_{i,4}$ the edges in the beginning of the head and $S_{i,5}$ the edges in the end of the head.   
Let us denote the number of edges in these subsets by $X_{i,1}, \ \ldots, \ X_{i,5}$.  

$X_{i,1}$ is at most the sum $H_{h_i}$ of the elements of the head, at most the sum $T_n - T_i$ of the elements of the 
tail, and at most the product $h_i(n - i)$ of the elements of the pairs formed from the tail and  from the 
beginning of the head, that is  
\begin{equation}
X_{i,1} \leq \min(H_{h_i},T_n - T_i,h_i(n - i)).\label{eq-HS1}
\end{equation}

A similar train of thought results 
\begin{equation}
X_{i,2} \leq \min(H_i - H_{h_i},T_n - T_i,(i - h_i)(n - i)).\label{eq-HS2}
\end{equation}

$X_{i,3}$ is at most $h_i(i - h_i)$ and at most $H_i$, implying 
\begin{equation}
X_{i,3} \leq \min(h_i(i - h_i),H_i).\label{eq-HS3} 
\end{equation}

$X_{i,4}$ is at most $\binom{h_i}{2}$ and at most $H_{h_i}$, implying 
\begin{equation}
X_{i,4} \leq \min(\binom{h_i}{2},H_{h_i}),\label{eq-HS4} 
\end{equation}
while $X_{i,5}$ is at most $\binom{i - h_i}{2}$ and at most $H_i - H_{h_i}$, implying 
\begin{equation}
X_{i,5} \leq \binom{i - h_i}{2}.\label{eq-HS5} 
\end{equation}
A requirement is also, that the tail can overrun its capacity, that is  
\begin{equation}
X_{i,1} + X_{i,2}  \leq T_i.\label{eq-HS6}
\end{equation}

Summing of \eqref{eq-HS1}, \eqref{eq-HS2}, \eqref{eq-HS3}, \eqref{eq-HS4}, and \eqref{eq-HS5} results 
\begin{equation}
H_i \leq X_{i,1} + X_{i,2} + X_{i,3} + 2X_{i,4} + 2X_{i,5}.\label{eq-HS7}
\end{equation}

Substituting of \eqref{eq-HS1}, \eqref{eq-HS2}, \eqref{eq-HS3}, \eqref{eq-HS4}, and \eqref{eq-HS5} into \eqref{eq-HS7} 
results \eqref{eq-HSa}, while \eqref{eq-HS6} is equivalent to \eqref{eq-HSb}.
\end{proof}

The following algorithm executes the test based on Lemma \ref{lemma-headsplitting}.

\textit{Input}. $n $: the number of vertices $(n \geq 1)$; \\
$b = (b_1, \ldots, b_n)$: an $n$-even sequence, accepted by \textsc{Binomial-Test}; \\
$H = (H_1,\ldots,H_n)$: $H_i$ the sum of the first $i$ elements of $b$; \\
$T = (T_1,\ldots,T_n)$: $T_i$ the sum of the last $n - i$ elements of $b$.

\textit{Output}. $L$: logical variable ($L = \textsc{False}$ signals,that $b$ is not graphical, while  
$L = \textsc{True}$ shows, that the test \textit{could nor decide}, whether $b$ is graphical).

\textit{Working variables}. $i$: cycle variable; \\
$h$: the actual value of $h_i$; \\
$X = (X_1,X_2,X_3,X_4,X_5)$: $X_j$ is the value of the actual $X_{i,j}$. 

\newpage
\noindent \textsc{Headsplitter-Test}$(n,b,H,T,L)$
\vspace{-2mm}
\begin{tabbing}%
199 \= xxx\=xxx\=xxx\=xxx\=xxx\=xxx\=xxx\=xxx \+ \kill 
\hspace{-7mm}01 \textbf{for} \= $i = 2$ \textbf{to} $n - 1$ \\
\hspace{-7mm}02              \> $h = \lfloor i/2 \rfloor$ \\
\hspace{-7mm}03              \> $X_1 = \min(H_h,T_n - T_i,h(n - i))$ \\
\hspace{-7mm}04              \> $X_2 = \min(H_i - H_h,T_n - T_i,(i - h)(n - i))$ \\
\hspace{-7mm}05              \> $X_3 = \min(h(i - h)$ \\
\hspace{-7mm}06              \> $X_4 = \binom{h}{2}$ \\
\hspace{-7mm}07              \> $X_5 = \binom{i - h}{2}$ \\
\hspace{-7mm}06              \> \textbf{if} \= $H_i > X_1 + X_2 + X_3 + 2X_4 + 2X_5$ or $X_1 + X_2 > T_i$ \\
\hspace{-7mm}07              \>             \> $L = \textsc{False}$ \\
\hspace{-7mm}08              \>             \> \textbf{return} $L$ \\
\hspace{-7mm}09 $L = \textsc{True}$ \\
\hspace{-7mm}10 \textbf{return} $L$
\end{tabbing}

The running time of the algorithm is $\Theta(1)$ in best, and $\Theta(n)$ in worst case.

It is a substantial circumstance that the use of Lemma \ref{lemma-binom} and  Lemma \ref{lemma-headsplitting}  requires 
only \textit{linear} time (while the earlier two theorems require quadratic time). But these improvements of Erd\H os-Gallai theorem  
decrease only the coefficient of the quadratic member in the formula of the running time, the order of growth remains unchanged.

Figure \ref{fig-numberofseq} contains the results of the running of \textsc{Binomial-Test} and \textsc{Head-} \linebreak 
\noindent \textsc{splitter-Test}, 
further the values $G(n)$ and $\frac{G(n)}{G(n + 1)}$ (the computation of the values of the function $G(n)$ will be 
explained in Section \ref{sec-gamma}).

Figure \ref{fig-apprrelativ} shows the relative frequency of the zerofree regular, binomial, headsplitted and graphical sequences 
compared to the number of regular sequences.

\subsection{Composite test\label{subsec-composite}}
\textsc{Composite-Test} uses approximate algorithms in the following order: \textsc{Parity-Test}, 
\textsc{Binomial-Test}, \textsc{Positive-Test}, \textsc{Headsplitter-Test}. 

\bigskip
\noindent \textsc{Composite-test}$(n,b,L)$
\vspace{-2mm}
\begin{tabbing}%
199 \= xxx\=xxx\=xxx\=xxx\=xxx\=xxx\=xxx\=xxx \+ \kill 
\hspace{-7mm}01 \textsc{Parity-Test}$(n,b,L)$ \\
\hspace{-7mm}02 \textbf{if} \= $L == \textsc{False}$ \\
\hspace{-7mm}03             \> \textbf{return} $L$ \\
\hspace{-7mm}04 \textsc{Binomial-Test}$(n,b,H,L)$ \\
\hspace{-7mm}05 \textbf{if} \= $L == \textsc{False}$ \\
\hspace{-7mm}06             \> \textbf{return} $L$ \\
\hspace{-7mm}07 \textsc{Headsplitter-Test}$(n,b,H,T,L)$ \\
\hspace{-7mm}08 \textbf{if} \= $L == \textsc{False}$ \\
\hspace{-7mm}09             \> \textbf{return} $L$ \\
\hspace{-7mm}10 $L = \textsc{True}$ \\  
\hspace{-7mm}11  \textbf{return} $L$
\end{tabbing}

The running time of this composite algorithm is in all cases $\Theta(n)$.

\begin{figure}[!t]
\begin{small}
\begin{center}
\begin{tabular}{||r|r|r|r|r||}  \hline \hline
$n$  & $B_z(n)$      &$F_z(n)$           & $G(n)$  & $G(n+1)/G(n)$    \\ \hline \hline
$1$  &$\mathbf{1}$ &$\mathbf{0}$     & $\mathbf{1}$ & $2.000000$\\ \hline
$2$  &$\mathbf{2}$ &$\mathbf{2}$     & $\mathbf{2}$ & $2.000000$\\ \hline
$3$  &$\mathbf{4}$ &$\mathbf{4}$     & $\mathbf{4}$ & $2.750000$      \\ \hline
$4$  &$\mathbf{11}$ &$\mathbf{11}$    & $\mathbf{11}$ & $2.818182$\\ \hline
$5$  &$\mathbf{31}$&$\mathbf{31}$    &$\mathbf{31}$  & $3.290323$     \\ \hline
$6$  &    $103$    &$\mathbf{102}$   &$\mathbf{102}$ & $3.352941$      \\ \hline 
$7$  &    $349$    &$344$            &$\mathbf{342}$ & $3.546784$    \\ \hline 
$8$  &    $1256$   &$1230$           &   $1213$      & $3.595218$\\ \hline 
$9$  &    $4577$   &$4468$           &  $4361$       & $3.672552$ \\ \hline 
$10$ &     $17040$ &$16582$          &$16016$        & $3.705544$\\ \hline 
$11$ &     $63944$ &$62070$          &$59348$        & $3.742620$ \\ \hline 
$12$ &   $242218$  &$234596$         &$222117$       & $3.765200$ \\ \hline 
$13$ & $922369$    &$891852$         & $836315$      & $3.786674$ \\ \hline 
$14$ &$3530534$    &$3409109$        & $3166852$     & $3.802710$ \\ \hline 
$15$ &$13563764$   &$13082900$       &   $12042620$  & $3.817067$ \\ \hline 
$16$ &$52283429$   &$50380684$       &   $45967479$  & $3.828918$ \\ \hline 
$17$ &$202075949$  &$194550002$      & $176005709$   & $3.839418$ \\ \hline 
$18$ &$782879161$  &$753107537$      &$675759564$   & $3.848517$ \\ \hline 
$19$ &$3039168331$ &$2921395019$     & $2600672458$  & $3.856630$ \\ \hline 
$20$ &$11819351967$&$11353359464$    & $10029832754$ & $3.863844$  \\ \hline 
$21$ &             &                 & $38753710486$ & $3.870343$ \\ \hline 
$22$ &             &                 &$149990133774$       & $3.876212$ \\ \hline 
$23$ &             &                 &$581393603996$       & $3.881553$\\ \hline 
$24$ &             &                 &$2 256 710 139 346$  & $3.886431$        \\ \hline 
$25$ &             &                 &$8 770 547 818 956$  & $3.890907$      \\ \hline 
$26$ &             &                 &$34 125 389 919 850$ & $3.895031$\\ \hline 
$27$ &             &                 &$132 919 443 189 544$& $3.897978$        \\ \hline 
$28$ &             &                 &$518 232 001 761 434$& $3.898843$      \\ \hline
$29$ &             &                 &$2 022 337 118 015 338$&           \\ \hline \hline
\end{tabular}
\end{center}
\end{small}
\caption{Number of zerofree binomial, zerofree headsplitted and graphical sequences, further the ratio of the numbers 
of graphical sequences for neighbouring values of $n$\label{fig-numberofseq}}
\end{figure}

\begin{figure}[!t]
\begin{small}
\begin{center}
\begin{tabular}{||c|r|c|c|c|c||}  \hline \hline
$n$ &   $E_z(n)$           & $E_z(n)/R(n)$       & $B_z(n)/R(n)$       &$F_z(n)/R(n)$         &$G(n)/R(n)$ \\ \hline \hline
$1$ &$0$               &$0.000000$         &$1.000000$         &$1.000000$          &$1.000000$  \\ \hline
$2$ &$1$               &$0.333333$         &$0.666667$         &$0.666667$          &$0.666667$ \\ \hline
$3$ &$2$               &$0.300000$         &$0.400000$         &$0.400000$          &$0.400000$ \\ \hline
$4$ &$9$               &$0.257143$         &$0.314286$         &$0.314286$          &$0.314286$ \\ \hline
$5$ &$28$              &$0.230159$         &$0.246032$         &$0.246031$          &$0.246032$ \\ \hline
$6$ &$110$             &$0.238095$         &$0.222943$         &$0.220779$          &$0.220779$ \\ \hline
$7$ &$396$             &$0.231352$         &$0.203380$         &$0.200466$          &$0.199301$ \\ \hline 
$8$ &$1 519$           &$0.236053$         &$0.195183$         &$0.191142$          &$0.188500$ \\ \hline
$9$ &$5 720$           &$0.235335$         &$0.188276$         &$0.183793$          &$0.179391$ \\ \hline
$10$&$21 942$          &$0.237524$         &$0.184460$         &$0.179502$          &$0.173375$ \\ \hline 
$11$&$83 980$          &$0.238098$         &$0.181290$         &$0.175977$          &$0.168260$ \\ \hline
$12$&$323 554$         &$0.239301$         &$0.179145$         &$0.173508$          &$0.164278$ \\ \hline
$13$&$1 248 072$       &$0.240000$         &$0.177368$         &$0.171500$          &$0.160821$ \\ \hline
$14$&$4 829 708$       &$0.240784$         &$0.176014$         &$0.169960$          &$0.157882$ \\ \hline
$15$&$18 721 080$      &$0.241379$         &$0.174884$         &$0.168684$          &$0.155271$ \\ \hline
$16$&$72 714 555$      &$0.241946$         &$0.173965$         &$0.167634$          &$0.152950$ \\ \hline
$17$&$282 861 360$     &$0.242424$         &$0.173188$         &$0.166738$          &$0.150844$ \\ \hline
$18$&$1 101 992 870$   &$0.242860$         &$0.172533$         &$0.165972$          &$0.148926$ \\ \hline
$19$&$4 298 748 300$   &$0.243243$         &$0.171970$         &$0.165306$          &$0.147158$ \\ \hline
$20$&$16 789 046 494$  &$0.243590$         &$0.171486$         &$0.164725$          &$0.145521$ \\ \hline
$21$&                  &                   &                   &                    &$0.143997$ \\ \hline
$22$&                  &                   &                   &                    &$0.142569$ \\ \hline
$23$&                  &                   &                   &                    &$0.141228$ \\ \hline
$24$&                  &                   &                   &                    &$0.139961$ \\ \hline
$25$&                  &                   &                   &                    &$0.138762$ \\ \hline 
$26$&                  &                   &                   &                    &$0.137625$ \\ \hline
$27$&                  &                   &                   &                    &$0.136542$ \\ \hline
$28$&                  &                   &                   &                    &$0.135509$ \\ \hline
$29$&                  &                   &                   &                    &$0.134521$ \\ \hline \hline
\end{tabular}
\end{center}
\end{small} 
\caption{The number of zerofree even sequences, further the ratio of the numbers binomial/regular, headsplitted/regular and graphical/regular sequences\label{fig-apprrelativ}}
\end{figure}

\section{Properties of the approximate testing algorithms\label{sec-approx}}
We investigate the efficiency of the approximate algorithms testing the regular algorithms.
Figure \ref{fig-regevenratio} contains the number $R(n)$ of regular, the number $E(n)$ of 
even, and the number $G(n)$ of graphical sequences for $n = 1, \ \ldots, \ 38$.  

The relative efficiency of arbitrary testing algorithm A for sequences of given length $n$ we define with the ratio of the number of accepted 
by A sequences of length $n$ and the number of graphical sequences $G(n)$. This ratio as a function of $n$ 
will be noted by $X_A(n)$ and called the \textit{error function} of A \cite{Ivanyi2011AoInf}. 

We investigate the following approximate algorithms, which are the components of \textsc{Composite-Test}:

1) \textsc{Parity-Test};

2) \textsc{Binomial-Test};

3) \textsc{Headsplitter-Test}.

According to \eqref{eq-lum3} there are $R(2) = 3$ 2-regular sequences: $(1,1), \ (1,0)$ and  
$(0,0).$ According to \eqref{eq-En}  among these sequences there are $E(2) = 2$ even sequences. 
\textsc{Binomial-Test} accepts both even ones, 
therefore $B(2) = 2.$ Both sequences are 2-graphical, therefore $G(2) = 2$ and  so the efficiency of \textsc{Parity-Test} (PT)
and \textsc{Binomial-Test} (BT) is $\textsc{X}_{\textsc{PT}}(2) = \textsc{X}_{\textsc{BT}}(2) = 2/2 = 1$, 
in this case both algorithms are optimal. 

The number of $3$-regular sequences is $R(3) = 10.$ From these sequences $(2,2,2), \ (2,2,0), \ (2,1,1), \ (2,0,0) \ (1,1,0)$ 
and $(0,0,0)$ are even, so $E(3) = 6.$ \textsc{Binomial-Test} excludes the sequences $(2,2,0)$ and 
$(2,0,0),$ so remains $B(3) = 4.$ Since these sequences are $3$-graphical, $G(3) = 4$ implies  
$\textsc{X}_{\textsc{PT}}(3) = \frac{3}{2}$ and $\textsc{X}_{\textsc{BT}}(3) = 1.$

The number of $4$-regular sequences equals to $R(4) = 35.$ From these sequences 16 is even, and the following 11 
are $4$-graphical: $(3,3,3,3)$, $(3,3,2,2)$, $(3,2,2,1)$, $(3,1,1,1,)$, $(2,2,2,2)$, $(2,2,2,0)$, $(2,2,1,1)$, $(2,1,1,0)$, 
$(1,1,1,1)$, $(1,1,0,0)$ and $(0,0,0,0)$. From the 16 even sequences \textsc{Binomial-Test} also excludes the 5 sequences, 
so  $B(4) = G(4) = 11$ and $X_{\textsc{BT}}(4)$ = 1.

According to these data in the case of $n \leq 4$ \textsc{Binomial-Test} recognizes all nongraphical sequences. 
Figure \ref{fig-numberofseq} shows, that for $n \leq 5$ we have $B(n) = G(n)$, that is \textsc{Binomial-Test} accepts 
the same number of sequences as the precise algorithms. If $n > 5$, then the error function of \textsc{Binomial-Test} 
is increasing: while $X_{\textsc{BT}}(6) = \frac{103}{102}$ (BT accepts one nongraphical sequence), 
$X_{\textsc{BT}}(7) = \frac{349}{342}$ (BT accepts 7 nongraphical sequences) etc.  

Figure \ref{fig-testtime} presents the average running time of the testing algorithms BT and HT 
in secundum and in number of operations. The data contain the time and operations necessary for the generation of the 
sequences too.

\begin{figure}[!ht]
\begin{footnotesize}
\begin{center}
\begin{tabular}{||r|r|r|r|r||}  \hline \hline
$n$        & BT, s              & BT, operation     & HT, s              & HT, operation   \\ \hline
1          &         0          &         14        &         0          &         15        \\
2          &         0          &         41        &         0          &         43        \\
3          &         0          &         180       &         0          &         200      \\
4          &         0          &         716       &         0          &         815      \\
5          &         0          &         2 918     &         0          &         3 321   \\
6          &         0          &         11 918    &         0          &         13 675 \\
7          &         0          &         48 952    &         0          &         56 299 \\
8          &         0          &         201 734   &         0          &         233 182           \\
9          &         0          &         831 374   &         0          &         964 121           \\
10         &         0          &         3 426 742 &         0          &         3 988 542        \\
11         &         0          &     14 107 824    &         0          &         16 469 036      \\
12         &         0          &   58 028 152      &         0          &         67 929 342      \\
13         &         0          &    238 379 872    &         0          &         279 722 127    \\
14         &         0          &    978 194 400    &         1          &         1 150 355 240 \\
15         &         2          &     4 009 507 932 &         3          &         4 724 364 716 \\
16         &         6          & 16 417 793 698    &         13         &         19 379 236 737          \\
17         &         26         & 67 160 771 570    &         51         &         79 402 358 497          \\
18         &        106         & 274 490 902 862   &         196        &         324 997 910 595        \\
19         &       423          & 1 120 923 466 932 &         798        &         1 328 948 863 507     \\
20         &      1 627         & 4 573 895 421 484 &         3 201      &         5 429 385 115 097     \\ \hline \hline
\end{tabular}
\end{center} 
\end{footnotesize}
\caption{Running time of \textsc{Binomial-Test} (BT) and \textsc{Headsplitter-Test} (HT) in secundum and 
as the number of operations for $n = 1, \ \ldots, \ 20$\label{fig-testtime}}
\end{figure}

\section{New precise algorithms\label{sec-newprecise}}
In this section the zerofree algorithms, the shifting Havel-Hakimi, the parity checking Havel-Hakimi, 
the shortened Erd\H os-Gallai, the jumping Erd\H os-Gallai, the linear Erd\H os-Gallai 
and the quick Erd\H os-Gallai algorithms are presented.

\subsection{Zerofree algorithms \label{subsec-zerofree}}
Since the zeros at the and of the input sequences correspond to isolated vertices, so they have no influence on the quality 
of the sequence. This observation is exploited in the following assertion, in which $p$ means the number of the 
positive elements of the input sequence. 

\begin{corollary} If\label{cor-zerofree} $n \geq 1$, the $(b_1,\ldots,b_n)$ $n$-regular   
sequence is $n$-graphical if and only if $(b_1, \ldots, b_p)$ is $p$-graphical.
\end{corollary}

\begin{proof} If all elements of $b$ are positive (that is $p = n$), then the assertion is equivalent with Erd\H os-Gallai theorem. 
If $b$ contains zero element (that is $p < n$), then the assertion is the consequence of Havel-Hakimi and Erd\H os-Gallai algorithms, 
since the zero elements do not help in the pairing of the positive elements, but from the other side they have no 
own requirement. 
\end{proof}

The algorithms based on this corollary are called \textsc{Havel-Hakimi-Zerofree (HHZ)}, resp. 
\textsc{Erd\H os-Gallai-Zerofree (EGZ)}.  

\subsection{Shifting Havel-Hakimi algorithm\label{subsec-HHShi}}
The natural algorithmic equivalent of the original Havel-Hakimi theorem is called \textsc{Havel-Hakimi Sorting} (HHSo), 
since it requires the sorting of the reduced input sequence in every round. 

But it is possible to design such implementation, in which the reduction of the degrees is executed saving 
the monotonity of the sequence. Then we get \textsc{Havel-Hakimi-Shifting} (HHSh) algorithm.

For the pseudocode of this algorithms see \cite{IvanyiLMS2011AML}.

\subsection{Parity checking Havel-Hakimi algorithm\label{subsec-HHP}}
It is an interesting idea the join the application of the conditions of Erd\H os-Gallai and Havel-Hakimi theorems in such a manner, 
that we start with the parity checking of the input sequence, and only then use the recursive Havel-Hakimi method.  

For the pseudocode of the algorithm \textsc{Havel-Hakimi-Parity} (HHP) see  \cite{IvanyiLMS2011AML}.

\subsection{Shortened Erd\H os-Gallai algorithm (\textsc{EGSh})\label{subsec-EGS}}
In the case of a regular sequence the maximal value of $H_i$ is $n(n - 1),$ therefore the inequality (\ref{eq-EGbinom}) 
 certainly holds for $i = n,$ therefore it is unnecessary to check.

Even more useful observation is contained in the following assertion due to Tripathi and Vijai.

\begin{lemma} \emph{(Tripathi, Vijay \cite{TripathiV2003})} If\label{lemma-EGS} $n \geq 1,$ then an $n$-regular sequence 
$b = (b_1\ldots,b_n)$ is $n$-graphical if and only if  
\begin{equation}
H_n \quad \mbox{even}\label{eq-EGSparos} 
\end{equation}
and
\begin{equation}
H_i  \leq \min(H_i,i(i - 1)) +\sum _{k=i+1}^n \min(i,b_k) 
\quad (i = 1,2,\ldots,r),\label{eq-EGStbinom}
\end{equation}
where 
\begin{equation}
r = \max_{1 \leq s \leq n}(s \ | \ s(s - 1) < H_s)
\end{equation}
\end{lemma}  

\begin{proof} If $i(i - 1) \geq H_i,$ then the left side of (\ref{eq-EGbinom}) is nonpositive, therefore 
the inequality holds, so the checking of the inequality is nonnecessary. 
\end{proof}

The algorithm based on this assertion is called \textsc{Erd\H os-Gallai-Shortened}. 
For example if the input sequence is $b = (5^{100})$, then \textsc{Erd\H os-Gallai} computes the right side of (\ref{eq-EGbinom}) 
99 times, while \textsc{Erd\H os-Gallai-Shortened} only 6 times. 

\subsection{Jumping Erd\H os-Gallai algorithm\label{subsec-EGJ}}
Contracting the repeated elements a regular sequence $(b_1,\ldots,b_n)$ can be written in the form 
$(b_{i_1}^{e_1}, \ldots, b_{i_q}^{e_q})$, where $b_{i_1} < \cdots < b_{i_q},$ $e_1, \ \ldots, \ e_q \geq 1$ and  
$e_1 + \cdots + e_q = n.$ Let $g_j = e_1 + \cdots + e_j \ (j = 1, \ \ldots, \ q).$   

The element $b_i$ is called the \textit{checking points} of the sequence $b$, if 
$i = n$ or $1 \leq i \leq n - 1$ és $b_i > b_{i+1}$. 
Then the checking points are $b_{g_1},\ldots,b_{g_q}$. 

\begin{theorem} \emph{(Tripathi, Vijay \cite{TripathiV2003})} An\label{lemma-EGJ} $n$-regular sequence $b = (b_1,\ldots,b_n)$ 
is $n$-graphical if and only if
\begin{equation} 
H_n \quad \mbox{even}\label{eq-TVparos}
\end{equation}
and
\begin{equation}
H_{g_i} - g_i(g_i - 1) \leq \sum _{k = g_i + 1}^n \min(i,b_k) \quad (i = 1,\ldots,q).\label{eq-TVbinom}
\end{equation}
\end{theorem}

\begin{proof} See \cite{TripathiV2003}.
\end{proof}

Later in algorithm \textsc{Erd\H os-Gallai-Enumerating} we will exploit, that in the inequality (\ref{eq-TVbinom}) 
$g_q$ is always $n,$ therefore it is enough to check the inequality only up to $i = q - 1$.  

The next program implements a quick version of Erd\H os-Gallai algorithm, exploiting Corollary \ref{cor-zerofree}, 
Lemma \ref{lemma-EGS} and  Lemma \ref{lemma-EGJ}.  
In this paper we use the pseudocode style proposed in \cite{CormenLRS2009}.

\textit{Input.} $n $: number of vertices $(n \geq 1)$; \\
$b = (b_1, \ldots, b_n)$: an $n$-even sequence.

\textit{Output.} $L$: logical variable ($L = \textsc{False}$ signalizes, that, $b$ is not graphical, while  
$L = \textsc{True}$ shows, that $b$ is graphical). 

\textit{Working variables.} $i$ and $j$: cycle variables; \\
$H = (H_0, H_1,\ldots,H_n)$: $H_i$ is the sum of the first $i$ elements of $b$; \\
$C$: the degree capacity of the actual tail; \\
$b_{n + 1}$: auxiliary variable helping to decide, whether $b_n$ is a jumping element.  

\bigskip
\noindent \textsc{Erd\H os-Gallai-Jumping}$(n,b,H,L)$
\vspace{-2mm}
\begin{tabbing}%
199 \= xxx\=xxx\=xxx\=xxx\=xxx\=xxx\=xxx\=xxx \+ \kill
\hspace{-7mm}01 $H_1 = b_1$                                 \` \textbf{//} lines 01--06: test of parity \\
\hspace{-7mm}02 \textbf{for} \= $i = 2$ \textbf{to} $n$ \\
\hspace{-7mm}03              \> $H_i = H_{i-1} + b_i$ \\
\hspace{-7mm}04 \textbf{if} \= $H_n$ odd \\
\hspace{-7mm}05             \> $L = \textsc{False}$ \\
\hspace{-7mm}06             \>  \textbf{return} $L$  \\
\hspace{-7mm}07 $b_{n + 1} = -1$                        \` \textbf{//} lines 07--20: test of the request of the head \\
\hspace{-7mm}08 $i = 1$ \\
\hspace{-7mm}09 \textbf{while} \= $i \leq n$ and $i(i - 1) < H_i$\\
\hspace{-7mm}10                \> \textbf{while} \= $b_i == b_{i+1}$ \\
\hspace{-7mm}11                \>                \> $i = i + 1$ \\  
\hspace{-7mm}12                \> $C = 0$ \\
\hspace{-7mm}13                \> \textbf{for} \= $j = i + 1$ \textbf{to} $n$ \\
\hspace{-7mm}14                \>              \> $C = C + \min(j,b_j)$ \\
\hspace{-7mm}15                \> \textbf{if} \= $H_i > i(i - 1) + C$ \\
\hspace{-7mm}16                \>             \> $L = \textsc{False}$ \\  
\hspace{-7mm}17                \>             \> \textbf{return} $L$ \\
\hspace{-7mm}18                \> $i = i + 1$ \\
\hspace{-7mm}19  $L = \textsc{True}$ \\
\hspace{-7mm}20 \textbf{return} $L$  
\end{tabbing}

The running time of \textsc{EGJ} varies between the best $\Theta(1)$ and the worst $\Theta(n^2)$.  

\subsection{Linear Erd\H os-Gallai algorithm\label{subsec-EGlinear}}
Recently we could improve \textsc{Erd\H os-Gallai} algorithm \cite{IvanyiL2011Comb,IvanyiLMS2011AML}. 
The new algorithm \textsc{Erd\H os-Gallai-Linear} exploits, that $b$ is monotone. It determines the capacities $C_i$  
in constant time. The base of the quick computation is the sequence $w(b)$ containing the \textit{weight points} $w_i$ 
of the elements of the input sequence $b$.
 
For given sequence $b$ let $w(b) = (w_1,\ldots,w_{n-1}),$ where $w_i$ gives the index of $b_k$ 
having the maximal index among such elements of $b$ which are greater or equal to $i.$

\begin{theorem} \emph{(Iv\'anyi, Lucz \cite{IvanyiL2011Comb}, Iv\'anyi, Lucz, M\'ori, S\'ot\'er \cite{IvanyiLMS2011AML})} 
If  $n \geq 1,$\label{theorem-EGL} then the $n$-regular sequence $(b_1, \ldots, b_n)$ is $n$-graphical if and only if
\begin{equation}
H_n \quad \mbox{is even} \label{eq-EGLparos}
\end{equation}
and if $i > w_i,$ then
\[
H_i \leq i(i - 1) + H_n - H_i,
\]
further if $i \leq w_i,$ then
\[
H_i \leq i(i - 1) + i(w_i -i) + H_n - H_{w_i}.
\]
\end{theorem}
\begin{proof}
(\ref{eq-EGLparos}) is the same as (\ref{eq-EGparity}). 

During the testing of the elements of $b$ by \textsc{Erd\H os-Gallai-Linear} there are two cases: 
\begin{itemize}
\item if $i > w_i,$ then the contribution   $C_i=\sum_{k=i+1}^n{\min(i,b_k)}$ of the tail of $b$ equals to 
$H_n - H_i,$ since the contribution $c_k$ of the element $b_k$ is only $b_k.$ 
\item if $i \leq w_1,$ then the contribution of the tail of $b$ consists of contributions of two types: $c_{i+1},\ldots,c_{w_i}$ 
are equal to $i,$ while $c_j = b_j$ for $j = w_{i}+ 1,\ldots,n.$ 
\end{itemize}

Therefore in the case $n - 1 \geq i > w_i$ we have
\begin{equation}
C_i =  H_n - H_i,\label{eq-EGLClargei}
\end{equation}
and in the case $1 \leq i \leq w_i$ 
\begin{equation}
C_i =  i(w_i -i) + H_n - H_{w_i}.\label{eq-EGLCsmalli}
\end{equation}
\end{proof}

The following program is based on Theorem \ref{theorem-EGL}. 
It decides on arbitrary $n$-regular sequence whether it is $n$-graphical or not. 

\textit{Input}. $n$: number of vertices $(n \geq 1)$; \\
$b = (b_1,\ldots,b_n)$: $n$-regular sequence.  

\textit{Output}. $L$: logical variable, whose value is \textsc{True}, if the input is graphical, 
and it is \textsc{False}, if the input is not graphical. 

\textit{Work variables.} $i$ and $j$: cycle variables; \\
$H = (H_1,\ldots,H_n)$: $H_i$ is the sum of the first $i$ elements of the tested $b$; \\
$b_0$: auxiliary element of the vector $b$ \\
$w = (w_1,\ldots,w_{n-1})$: $w_i$ is the weight point of $b_i$, that is the maximum of the indices of such elements 
of $b,$ which are not smaller than $i$; \\
$H_0 = 0$: help variable to compute the other elements of the sequence $H$; \\
$b_0 = n -1$: help variable to compute the elements of the sequence $w$.

\newpage
\noindent \textsc{Erd\H os-Gallai-Linear}$(n,b,L)$
\vspace{-2mm}
\begin{tabbing}%
199 \= xxx\=xxx\=xxx\=xxx\=xxx\=xxx\=xxx\=xxx \+ \kill
\hspace{-7mm}01 $H_0 = 0$                                            \` \textbf{//} line 01: initialization \\
\hspace{-7mm}02 \textbf{for} \= $i = 1$ \textbf{to} $n$               \` \textbf{//} lines 02--03: computation of the elements of $H$    \\     
\hspace{-7mm}03              \> $H_i = H_{i-1} + b_i$ \\ 
\hspace{-7mm}04 \textbf{if} \= $H_n$ \textrm{ odd}    \` \textbf{//} lines 04--06: test of the parity    \\ 
\hspace{-7mm}05             \> $L = \textsc{False}$ \\
\hspace{-7mm}06             \> \textbf{return} $L$ \\
\hspace{-7mm}07 $b_0 = n - 1$                             \` \textbf{//} line 07: initialization of a working variable \\ 
\hspace{-7mm}08 \textbf{for} \= $i = 1$ \textbf{to} $n$               \` \textbf{//} lines 08--12: computation of the weights   \\ 
\hspace{-7mm}09              \> \textbf{if} \= $b_i < b_{i-1}$  \\ 
\hspace{-7mm}10              \>             \> \textbf{for} \= $j = b_{i-1}$ \textbf{downto} $b_i + 1$ \\
\hspace{-7mm}11              \>             \>              \> $w_j = i - 1$ \\
\hspace{-7mm}12              \>             \> $w_{b_i} = i$ \\ 
\hspace{-7mm}13 \textbf{for} \= $j = b_n$ \textbf{downto} $1$  \` \textbf{//} lines 13--14: large weights\\
\hspace{-7mm}14              \> $w_j = n$ \\  
\hspace{-7mm}15 \textbf{for} \= $i = 1$ \textbf{to} $n$            \` \textbf{//} lines 15--23: test of the elements of $b$ \\                          
\hspace{-7mm}16              \> \textbf{if} \= $i \leq w_i$        \` \textbf{//} lines 16--19: test of indices for large $w_i$'s \\ 
\hspace{-7mm}17              \>             \> \textbf{if} \= $H_i > i(i - 1) + i(w_i - i) + H_n - H_{w_i}$ \\
\hspace{-7mm}18              \>             \>             \> $L = \textsc{False}$ \\
\hspace{-7mm}19              \>             \>             \> \textbf{return} $L$ \\  
\hspace{-7mm}20              \> \textbf{if} \= $i > w_i$  \` \textbf{//} lines 20--23: test of indices for small $w_i$'s \\  
\hspace{-7mm}21              \>             \> \textbf{if} \= $H_i >  i(i - 1) + H_n - H_i$ \\
\hspace{-7mm}22              \>             \>             \> $L = \textsc{False}$ \\
\hspace{-7mm}23              \>             \>             \> \textbf{return} $L$ \\           
\hspace{-7mm}24 $L = \textsc{True}$ \` \textbf{//} lines 24--25: the program ends with the value \textsc{True} \\
\hspace{-7mm}25  \textbf{return} $L$  \\ 
\end{tabbing}

\begin{theorem} \emph{(Iv\'anyi, Lucz \cite{IvanyiL2011Comb}, Iv\'anyi, Lucz, M\'ori, S\'ot\'er \cite{IvanyiLMS2011AML})} 
Algorithm \textsc{Erd\H os-Gallai-Linear} decides in $\Theta(n)$ time, whether an  $n$-regular sequence 
$b = (b_1,\ldots,b_n)$ is graphical or not.  
\end{theorem} 

\begin{proof} Line 1 requires $O(1)$ time, lines 2--3 $\Theta(n)$ time, lines 4--6 $O(1)$ time, line 07 $O(1)$ time, 
lines 08--12 $O(1)$ time, lines 13--14 $O(n)$ time, lines 15--23 $O(n)$ time and lines 24--25 $O(1)$ time, 
therefore the total time requirement of the algorithm is $\Theta(n)$. 
\end{proof}

Since in the case of a graphical sequence all elements of the investigated sequence are to be tested, 
in the case of RAM model of computations \cite{CormenLRS2009} \textsc{Erd\H os-Gallai-Linear} is asymptotically optimal.

\section{Running time of the precise testing algorithms\label{sec-classrun}}
We tested the precise  algorithms determining their total running time for all the even sequences. 
The set of the even sequences is the smallest such set of sequences, whose the cardinality we know exact and explicite formula. 
The number of $n$-bounded sequences $K(n)$ is also known, but this function grows too quickly when $n$ grows. 

If we would know the average running time of the bounded sequences we would take into account that is is sufficient to 
weight the running times of the regular sequences with the corresponding frequencies. For example a homogeneous sequence 
consisting of identical elements would get a unit weight since it corresponds to only one bounded sequence, 
while  a rainbow sequence consisting is $n$ different elements as e.g. the sequence $n,n-1,\ldots,1$ corresponds 
to $n!$ different bounded sequences and therefore would get a corresponding weight equal to $n!$.

We follow two ways of the decreasing of the running time of the precise algorithms. 
The first way is the decreasing of the number of the executable operations. The second way is, that we try to use 
quick (linear time) preprocessing algorithms for the filtering of the sequences in order to decrease of the part of sequences 
requiring the relative slow precise algorithms.   

For the first type of decrease of the expected running time is the shortening of the sequences and the application of the 
checking points, while for the the second type are examples the completion of HH algorithm with the parity checking or 
the completion of the EG algorithm with the binomial and headsplitted algorithms.  

\begin{figure}[!t]
\begin{footnotesize}
\begin{center}
\begin{tabular}{||r|r|r|r|r|r||}  \hline \hline
$n$  &      HHSo &  HHSh   &   EG      &       EGJ    &  EGL        \\ \hline \hline
 1	 &	    10	 &	15	    &   87	   &        -     &  -	\\
 2	 &	    40	 &	61	    &  119	   &        12    &  37	\\
 3	 &	   231	 &	236	    &  267	   &       116    & 148	\\
 4	 &	1 170	 &	1 052   &  946	   &       551    & 585	\\
 5	 &	5 969	 &	4 477   & 4 000	   &     2 677    & 2 339	\\
 6	 &	31 121	 &	20 153  & 18 206	   &    12 068    & 9 539	\\
 7	 &	157 345	 &	88 548  & 82 154	   &    54 184    & 38 984	\\
 8	 &	784 341	 &  393 361   & 372 363   &    238 813   & 160 126	\\
 9	 &3 628 914  & 1 726 484 &1 666 167  &  1 666 167   &  656 575	\\
10	 &17 345 700 &7 564 112  &7 418 447  & 4 552 276    & 2 692 240	\\
11	 &80 815 538 &32 895 244 &32 737 155 & 19 680 986   &  11 018 710	\\
12	&385 546 527 &142 460 352&143 621 072&84 608 529   &  45 049 862	\\
13	&1 740 003 588&613 739 913&626 050 861& 362 141 061&  183 917 288	\\
14	&8 066 861 973&	2 633 446 908&2 715 026 827&1 543 745 902 &750 029 671	\\
15	&36 630 285 216&11 254 655 388&11 717 017 238&6 557 902 712 & 3 055 289 271	\\ \hline \hline
\end{tabular}
\end{center} 
\end{footnotesize}
\caption{Total number of operations as the function of $n$ for precise algorithms 
HHSo, HHSh, EG, EGJ, and EGL.\label{fig-precisetime}} 
\end{figure}
\begin{figure}[!t]
\begin{footnotesize}
\begin{center}
\begin{tabular}{||r|r|r|r|r|r||}  \hline \hline
$n$  &  $E(n)$  & $T(n)$, s  & $Op(n)$ & $T(n)/E(n)/n,$ s & $Op(n)/E(n)/n$       \\ \hline \hline   
2    &     2    &      0      &     37    &         0         &         9.25000000000     \\
3    &     6    &      0      &    148    &         0         &         8.22222222222    \\
4    &    19    &      0      &    585    &         0         &         7.69736842105    \\
5    &    66    &      0      &  2 339   &         0          &         7.08787878788   \\
6    &   236    &      0      &  9 539   &         0          &         6.73658192090   \\
7    &   868    &      0      & 38 984    &         0         &         6.41606319947    \\
8    &  3 235   &      0      &160 126    &         0         &         6.18724884080    \\
9    &   12 190 &      0      & 656 575   &         0         &         5.98464132714    \\
10   &   46 252 &      0      & 2 692 240 &         0         &         5.82080774885     \\
11   &  176 484 &      0      &11 018 710 &         0         &         5.67587378511  \\
12   &  676 270 &      0      &45 049 862 &         0         &         5.55126675243   \\
13   &2 600 612 &      0      &183 917 288&         0         &         5.44005937537    \\
14   &10 030 008&      1      &750 029 671&0.000000007121487  &         5.34132654018    \\
15   &38 781 096&      5      &3 055 289 271&0.000000008595253&         5.25219687963   \\
16   &150 273 315&    23      &12 434 367 770&0.000000009565903&        5.17156346504    \\
17   &583 407 990&    79      &50 561 399 261&0.000000007965367&        5.09797604337    \\
18   &2 268 795 980&  297     &205 439 740 365&0.00000000727258&        5.03056202928    \\ \hline \hline
\end{tabular}
\end{center} 
\end{footnotesize}
\caption{Total and amortized running time of \textsc{Erd\H os-Gallai-Linear} in secundum, resp. in the number 
of executed operations\label{fig-EGL}}
\end{figure}

In this section we investigate the following precise algorithms:

1) \textsc{Havel-Hakimi-Shorting} (HHSo).

2) \textsc{Havel-Hakimi-Shifting} (HHSh).

3) \textsc{Erd\H os-Gallai} algorithm (EG).

4) \textsc{Erd\H os-Gallai-Jumping} algorithm (EGJ).

5) \textsc{Erd\H os-Gallai-Linear} algorithm (EGL).

\noindent Figure \ref{fig-precisetime} contains the total number of operations of the algorithms HHSo, HHSh, EG, and EGL 
required for the testing of all even sequences of length $n = 1,\ldots, 15$. The operations necessary to generate the sequences 
are included.

Comparison of the first two columns shows that algorithm HHSh is much quicker than HHSo, especially if $n$ increases. Comparison 
of the third and fourth columns shows that we get substantial decrease of the running time if we have to test 
the input sequence only in the check points. Finally  the comparison of the third and fifth columns demonstrates the 
advantages of a linear algorithm over a quadratic one. 

Figure \ref{fig-EGL} shows the running time of \textsc{Erd\H os-Gallai-Linear} in secundum and operation, and also 
the amortized number of operation/even sequence.

\begin{figure}[!t]
\begin{footnotesize}
\begin{center}
\begin{tabular}{||r|r|r|r|r|r|r|r|r|r||}  \hline \hline
$E(n)-G(n)$&$n/i$    &$f_1$     & $f_2$     &  $f_3$&   $f_4$& $f_5$   & $f_6$ & $f_7$ \\ \hline \hline
  $2$     &  $3$     &$2$       & $0$       &  $0$  &   $0$  &   $0$   &  $0$  & $0$   \\
  $8$     &  $4$     &$6$       & $2$     &   $0$   &   $0$  &   $0$   &  $0$  & $0$  \\
  $35$    & $5$      &$33$      &   2     &  0      & 0      & 0       &  $0$  & $0$   \\
 $134$    & $6$      &$122$     &  $12$   & $0$     &   0    &   0     &  0    &   0   \\
$526 $    &  $7$     &$459$     &$65$     & $2$     & $2$    & 0       & 0     &  0     \\
$2022$    & $8$      &$1709$    &$289$    &$24$     & $0$    &  $0$    &  0    & 0     \\
$7829$    & $9$       &$6421$    &$1228$   &$176$    & $4$    &   $0$  &  0     & 0      \\
$30236$   & $10$      &$24205$   &$4951$   &$1013$   &$67$    &   $0$  &  0     & 0       \\
$115136$  & $11$      &$91786$   &$19603$  &$5126$   &  $610$ &   $11$ &  0     & 0       \\
$454153$  & $12$      &$349502$  &$76414$  &$23755$  & $4274$ &  $208$ &  0     & 0       \\
$1764297$ & $13$      &$1336491$ &$296036$ &$104171$ &$25293$ & $2277$ &  29    &$0$    \\
$6863156$ & $14$      &$5128246$ &$1142470$&$439155$ &$133946$&$18673$ & $666$  &$0$     \\
$26738476$& $15$      &$19739076$&$4404813$&$1803496$&$655291$&$127116$& $8603$ &$81$   \\ \hline \hline
\end{tabular}
\end{center} 
\end{footnotesize}
\caption{Distribution of the even nongraphical sequences according to the number of tests made by \textsc{Erd\H os-Gallai-Jumping} 
to exclude the given sequence for $n = 3, \ \ldots, \ 15$\label{fig-EGJrounds}}
\end{figure}

The most interesting data of Figure \ref{fig-EGL} are in the last column: they show that the number of 
operations/investigated sequence/length of the investigated sequence is monotone decreasing (see \cite{RuskeyCES1994}).

Figure \ref{fig-EGJrounds} shows the distribution of the $E(n) - G(n)$ even nongraphical sequences according to 
the number of tests made by \textsc{Erd\H os-Gallai-Jumping} to exclude the given sequence for  
$n = 3, \ \ldots, \ 15$ vertices. $f_i(n) = f_i$ gives the frequency of even nongraphical sequences of length $n$, 
which requeired exactly $i$ round of the test.

These data show, that the maximal number of tests is about $\frac{n}{2}$ in all lines.

Figure \ref{fig-EGJavrounds} shows the average number of required rounds for the nongraphical, graphical and all even 
sequences. The data of the column belonging to $G(n)$ are computed using Lemma \ref{cor-regrain}. It is remarkable that 
the sequences of the coefficients are monotone decreasing in the last three columns. 

Figure \ref{fig-EGb1} presents the distribution of the graphical sequences according to their first element. 
These data help at the design of the algorithm \textsc{Erd\H os-Gallai-Enumerating} which computes the new values of $G(n)$ 
(in the slicing of the computations belonging to a given value of $n$).

\begin{figure}[!t]
\begin{small}
\begin{center}
\begin{tabular}{||r|r|r|r|c|c|c||}  \hline \hline
$n$       & $E(n)$    &$G(n)$    &  $E(n) - G(n)$ & average of  & average of   & average of          \\  
          &           &          &                &$E(n) - G(n)$& $G(n)$       & $E(n)$   \\ \hline \hline
$3$       &  $6$      &  $4$     &    $2$         &$0.3333n$ & $0.8000n$   & $0.6444n$ \\ \hline
$4$       &  $19$     &  $11$    &    $8$         &$0.3125n$ & $0.5714n$   & $0.4661n$    \\  \hline
$5$       &  $66$     &  $31$    &     $35$       &$0.2114n$ & $0.5555n$   & $0.3730n$ \\ \hline
$6$       &  $236$    &  $102$   &     $134$      &$0.1967n$ & $0.5455n$   & $0.3730n$    \\ \hline
$7$       &  $868$    &  $342$   &     $526$      &$0.1649n$ & $0.5385n$   & $0.3475n$    \\ \hline
$8$       &  $3233$   &  $1213$  &     $2020$     &$0.1458n$ & $0.5333n$   & $0.2911n$ \\ \hline
$9$       & $12190$   & $4363$   &     $7829$     &$0.1337n$ & $0.5294n$   & $0.2753n$ \\ \hline
$10$      &$46232$    & $16016$  &     $30216$    &$0.1249n$ & $0.5263n$   & $0.2700n$       \\ \hline
$11$      &$174484$   & $59348$  &     $115136$   &$0.1175n$ & $0.5238n$   & $0.2557n$       \\ \hline
$12$      &$676270$   & $222117$ &     $454153$   &$0.1085n$ & $0.5217n$   & $0.2444n$ \\ \hline
$13$      &$2603612$  & $836313$ &     $1767299$  &$0.1035n$ & $0.5200n$   & $0.2373n$ \\ \hline
$14$      &$10030008$ &$3166852$ &   $6863156$    &$0.0960n$ & $0.5185n$   & $0.2294n$       \\ \hline
$15$      &$38761096$ &$12042620$&  $26718476$    &$0.0934n$ & $0.5172n$   & $0.2251n$  \\ \hline \hline
\end{tabular}
\end{center} 
\end{small}
\caption{Weighted average number of tests made by \textsc{Erd\H os-Gallai-Jumping} while investigating  the even sequences for $n = 3, \ \ldots, \ 15$\label{fig-EGJavrounds}}
\end{figure}
\begin{figure}[!t]
\begin{scriptsize}
\begin{center}
\begin{tabular}{||r|r|r|r|r|r|r|r|r|r|r|r|r|r|r|r|r||}  \hline \hline
$n/b_1$  & $0$   &$1$  & $2$ & $3$ & $4$ & $5$ & $6$ & $7$ & $8$ &$9$ & $10$ & $11$   \\ \hline \hline
$1$      &1      &     &     &     &     &     &     &     &     &    &      &           \\ \hline
$2$      &1      &   1 &     &     &     &     &     &     &     &     &     &          \\ \hline
$3$      &   1   &   1 &  2  &     &     &     &     &     &     &     &     &  \\ \hline
$4$      &   1   &   1 &  4  &  4  &     &     &     &     &     &     &     &   \\ \hline
$5$      &   1   &   2 &  7  &10   &  11 &     &     &     &     &     &     &    \\ \hline
$6$      &   1   &   3 &10   &  22 &  35 & 31  &     &     &     &     &     &    \\ \hline
$7$      &   1   &   3 &14   &  34 &  78 & 110 & 102 &     &     &     &     &    \\ \hline
$8$      &   1   &   4 & 18  &  54 & 138 & 267 & 389 & 342 &     &     &     &    \\ \hline
$9$      &   1   &   4 & 23  &  74 & 223 & 503 & 968 & 1352& 1213&     &     &     \\ \hline
$10$     &   1   &   5 & 28  & 104 & 333 & 866 & 1927& 3496& 4895&4361 &     &    \\ \hline
$11$     &   1   &   5 & 34  &134  &479  &1356 &3471 &7221 &12892&17793&16016&        \\ \hline 
$12$     &   1   &   6 & 40  &176  &661  &2049 &5591 &13270&27449&47757&65769&59348     \\ \hline \hline
\end{tabular}
\end{center} 
\end{scriptsize}
\caption{The distribution of the graphical sequences according to $b_1$ for $n = 1, \ \ldots, \ 12$\label{fig-EGb1}}
\end{figure}

\newpage
We see in Figure \ref{fig-EGb1} that from $n = 6$ the multiplicities increase up to $n - 2,$ and the last positive value is smaller 
then the last but one element.

\section{Enumerative results\label{sec-enum}}
Until now for example Avis and Fukuda \cite{AvisF1996}, Barnes and Savage \cite{BarnesS1995,BarnesS1997}, 
Burns \cite{Burns2007}, Erd\H os and Moser \cite{Moon1968}, 
Erd\H os and Richmond \cite{ErdosR1993}, Frank, Savage and Selers \cite{FrankSS2002}, Kleitman and Winston \cite{KleitmanW1981}, 
Kleitman and Wang \cite{KleitmanW1973}, Metropolis and Stein \cite{MetropolisS1980}, R{\o}dseth et al. \cite{RodsethST2009}, 
Ruskey et al. \cite{RuskeyCES1994}, Stanley \cite{Stanley1991}, Simion \cite{Simion1997} and Winston and Kleitman \cite{WinstonK1983} 
published results connected with the enumeration of degree sequences. Results connected with the number of sequences 
investigated by us can be found in the books of Sloane és Ploffe \cite{SloaneP1995}, further Stanley \cite{Stanley1997}  
and in the free online database \textit{On-line Encyclopedia of Integer Sequences}  
\cite{Sloane2011A001700,Sloane2011A004251,Sloane2011A005654} 

It is easy to show, that if $l, \ u$ and $m$ are integers, further $u \geq l,$ $m \geq 1$, and 
$l \leq b_i \leq u$ for $i = 1, \ldots, \ m,$ then 
the number of $(l,u,m)$-bounded sequences $a = (a_1,\ldots,a_m)$ of integer numbers $K(l,u,m)$ is 
\begin{equation}
K(l,u,m) = (u - l + 1)^m.\label{eq-lum1}
\end{equation}
 
It is known (e.g. see \cite[page 65]{Jarai2005}), that if $l, \ u$ and $m$ are integers, further $u \geq l$ and $m \geq 1$, 
and $u \geq b_1 \geq \cdots \geq b_n \geq l,$ then the number of $(l,u,m)$-regular sequences of integer numbers $R(l,u,m)$ is 
\begin{equation}
R(l,u,m) = \binom{u - l + m}{m}.\label{eq-lum2}
\end{equation}

The following two special cases of (\ref{eq-lum2}) are useful in the design of the algorithm \textsc{Erd\H os-Gallai-Enumerating}.  

If $n \geq 1$ is an integer, then the number of $R(0,n-1,n)$-regular sequences is 
\begin{equation}
R(0,n-1,n) = R(n) = \binom{2n - 1}{n}.\label{eq-lum3}
\end{equation}

If $n \geq 1$ is an integer, then the number of $R(1,n-1,n)$-regular sequences is 
\begin{equation}
R(1,n-1,n) = R_z(n) = \binom{2n - 2}{n}.\label{eq-lum4}
\end{equation}

In 1987 Ascher derived the following explicit formula for the number of $n$-even sequences $E(n)$.

\begin{lemma} \emph{(Ascher \cite{Ascher1987}, Sloane, Pfoffe \cite{SloaneP1995})} If \lemma{lemma-En} $n \geq 1,$ 
then the number of $n$-even sequences $E(n)$ is
\begin{equation}
E(n) = \frac{1}{2} \left ( \binom{2n -1}{n} + \binom{n - 1}{\lfloor n \rfloor} \right ).\label{eq-En}
\end{equation}
\end{lemma}

\begin{proof} See \cite{Ascher1987,SloaneP1995}.
\end{proof}

At the designing and analysis of the results of the simulation experiments is useful, if we know some features of the functions 
$R(n)$ and $E(n)$.

\begin{lemma} If $n \geq 1,$ then\label{lemma-Rnmon}
\begin{equation}
\frac{R(n + 2)}{R(n + 1)} > \frac{R(n + 1)}{R(n)},\label{eq-Rnmon}
\end{equation}
\begin{equation}
\lim _{n \rightarrow \infty} \frac{R(n + 1)}{R(n)} = 4,\label{eq-Rnlim}
\end{equation}
further
\begin{equation}
\frac{4^n}{\sqrt{4\pi n}}\left (1 - \frac{1}{2n} \right) < R(n) < 
\frac{4^n}{\sqrt{4\pi n}}\left(1 - \frac{1}{8n+8}\right).\label{Rn} 
\end{equation}
\end{lemma}

\begin{proof} 
On the base of \eqref{eq-lum3} we have
\begin{equation}
\frac{R(n + 2)}{R(n + 1)} = \frac{(2n + 3)!(n + 1)n!}{(n + 2)!(n + 1)!(2n + 1)!} 
 = \frac{4n + 6}{n + 2} = 4 - \frac{2}{n + 2}, 
\end{equation} 
from where we get directly (\ref{eq-Rnmon}) and (\ref{eq-Rnlim}).
\end{proof}

Using Lemma \ref{lemma-En} we can give the precise asymptotic order of growth of $E(n)$. 

\begin{lemma} \label{lemma-En} If  $n \geq 1,$ then 
\begin{equation}
\frac{E(n + 2)}{E(n + 1)} > \frac{E(n + 1)}{E(n)},\label{eq-Enmon}
\end{equation}
\begin{equation}
\lim _{n \rightarrow \infty} \frac{E(n + 1)}{E(n)} = 4,\label{eq-rhonlim}
\end{equation}
further
\begin{equation}
\frac{4^n}{\sqrt{\pi n}}(1 - D_3(n)) < E(n) < \frac{4^n}{\sqrt{\pi n}}(1 - D_4(n)),\label{eq-rhonass} 
\end{equation}
where $D_3(n)$ and $D_4(n)$ are monotone decreasing functions tending to zero.
\end{lemma}

\begin{proof} The proof is similar to the proof of Lemma \ref{lemma-Rnmon}.
\end{proof}

Comparison of (\ref{eq-lum3}) and Lemma \ref{lemma-En} shows, that the order of growth of numbers of 
even and odd sequences is the same, but there are more even sequences than odd. Figure \ref{fig-regevenratio} contains 
the values of $R(n)$, $E(n)$ and $E(n)/R(n)$ for $n = 1, \ \ldots, \ 37.$

As the next assertion and Figure \ref{fig-regevenratio} show, the sequence of the ratios $E(n)/R(n)$ 
is monotone decreasing and tends to $\frac{1}{2}$. 

\begin{corollary} If\label{cor-ER} $n \geq 1$, then 

\begin{equation}
\frac{E(n + 1)}{R(n + 1)} < \frac{E(n)}{R(n)}\label{eq-ERmonmon}
\end{equation}
and
\begin{equation}
\lim _{n \rightarrow \infty} \frac{E(n)}{R(n)} = \frac{1}{2}.\label{eq-ERlim}
\end{equation}
\end{corollary}

\begin{proof} This assertion is a direct consequence of \eqref{eq-lum3} and \eqref{eq-En}.

\end{proof}

The expected value of the number of jumping elements has a substantial influence on the running time of algorithms 
using the jumping elements. Therefore the following two assertions are useful.

The number of different elements in an $n$-bounded sequence $b$ is called \textit{the rainbow number} of the sequence, 
and it will be denoted by $r_n(b)$. 

\begin{lemma}\label{lemma-boundedrain}
Let $b$ be a random $n$-bounded sequence. Then the expectation and
 variance of its rainbow number are as follows.
\begin{align}
E[r_n(b)]  & = n\left[1-\left(1-\frac{1}{n}\right)^n\right] =
n\left(1-\frac{1}{e}\right)+O(1), \label{exp_rain}\\
Var[r_n(b)]& = n\left(1-\frac{1}{n}\right)^n
\left[1-\left(1-\frac{1}{n}\right)^n\right]\notag\\
& \hspace{1cm}+n(n-1)\left[\left(1-\frac{2}{n}\right)^n
-\left(1-\frac{1}{n}\right)^{2n}\right]\notag\\
& = \frac{n}{e}\left(1-\frac{2}{e}\right)+O(1). \label{var_rain}
\end{align}
\end{lemma}
\begin{proof}
Let $\xi _i$ denote the indicator of the event that number $i$
is not contained in a random $n$-bounded sequence. Then the rainbow
number of a random sequence is $n-\sum _{i=0}^{n-1}\xi _i$, hence its 
expectation equals $n-\sum _{i=0}^{n-1}E[\xi _i]$. Clearly,
\begin{equation}
E[\xi _i]=\left(1-\frac{1}{n}\right)^n
\end{equation}
holds independently of $i$, thus
\begin{equation}
E[r_n(b)] = n\left[1-\left(1-\frac{1}{n}\right)^n\right].
\end{equation}
On the other hand,
\begin{equation}
Var[r_n(b)] = Var\left[\sum _{i=0}^{n-1}\xi _i\right] =
\sum _{i=0}^{n-1}Var[\xi _i]+2\sum _{0\le i<j\le n-1}cov[\xi _i,\xi _j].
\end{equation}
Here
\begin{equation}
Var[\xi _i]=\left(1-\frac{1}{n}\right)^n \left[1-\left(1-\frac{1}{n}\right)^n\right],
\end{equation}
and
\begin{equation}
cov[\xi _i,\xi _j] = E[\xi _i\xi _j]-E[\xi _i]E[\xi _j] =
\left(1-\frac{2}{n}\right)^n-\left(1-\frac{1}{n}\right)^{2n},
\end{equation}
implying \eqref{var_rain}.
\end{proof}

We remark that this lemma answers a question of Imre K\'atai \cite{Katai2010} 
posed in connection with the speed of computers having interleaved memory 
and with checking algorithms of some puzzles (e.g sudoku).

\begin{lemma}\label{lemma-regrain}
The number of $(0,n-1,m)$-regular sequences composed from $k$ distinct numbers is
\begin{equation}
\dbinom{n}{k}\dbinom{m-1}{k},\ k=1,\dots,n.
\end{equation}
In other words, the distribution of the rainbow number $r_n(b)$ of a random $(0,n-1,m)$-regular sequence $b$ is hypergeometric 
with parameters $n+m-1$, $n$, and $m$.
\end{lemma}
\begin{proof}
The $k$-set of distinct elements of the sequence can be selected from $\{0,1,\dots,n-1\}$
in $\binom{n}{k}$ ways. Having this values selected we can tell their multiplicities in $\binom{m-1}{k-1}$ ways. 
Let us consider the $k$ blocks of identical elements. The first one starts with $b_1$, and the starting position of the
other $k-1$ blocks can be selected in $\binom{m-1}{k-1}$ ways.
\end{proof}

From this the expectation and the variance of a random $n$-regular sequence follow immediately.
\begin{corollary} Let\label{cor-regrain} $b$ be a random $n$-regular sequence. Then the expectation and the variance of its 
rainbow number $r_n(b)$ are as follows:
\begin{align}
E[r_n(b)]&=\frac{n^2}{2n-1} = \frac{n}{2} + \frac{n}{4n - 2} = \frac{n}{2} + O(1),\label{eq-regrainexp}\\
Var[r_n(b)]& = \frac{n^2(n-1)}{2(2n-1)^2} = \frac{n}{8} + \frac{n}{128n^2 - 128n + 32} = \frac{n}{8} + O(1).
\label{eq-regrainvar}
\end{align}
\end{corollary}

\begin{lemma} Let $b$ be a random $n$-regular sequence. Let us write it in the form $b
=(b_1^{e_1},\ldots,b_r^{e_r})$. Then the expected value of the exponents $e_j$ is   
\begin{equation}
E[e_j\mid r(b)\ge j] = 4 + o(1).\label{eq-runlength}
\end{equation}
\end{lemma}
\begin{proof}
Let $c(n,j)$ denote the number of $n$-regular sequences with
rainbow number not less than $j$. By Lemma \ref{lemma-regrain},
\begin{equation}\label{eq-jrun}
c(n,j)=\sum _{k=j}^n \binom{n}{k} \binom{n-1}{k-1}.
\end{equation}
Let us turn to the number of $n$-regular sequences with
rainbow number not less than $j$ and $e_j=\ell$. This is
equal to the number of $(0,n-1,n-\ell+1)$-regular sequences
containing at least $j$ different numbers, that is,
\begin{equation}
\sum_{k=j}^n\binom{n}{k}\binom{n-\ell}{k-1}.
\end{equation}
From this the sum of $e_j$ over all $n$-regular sequences
with $e_j>0$ is equal to
\begin{multline}
\sum_{\ell=1}^{n-j+1}\ell\sum_{k=j}^n\binom{n}{k}\binom{n-\ell}{k-1}=
\sum_{k=j}{n}\binom{n}{k}\sum_{\ell=1}^{n-j+1}\binom{\ell}{1}
\binom{n-\ell}{k-1}\\
=\sum_{k=j}^n\binom{n}{k}\binom{n+1}{k+1}=c(n+1,j+1).
\end{multline}
This can also be seen in a more direct way. Consider an 
arbitrary $n$-regular sequence with at least $j + 1$ blocks, 
then substitute the elements of the $j + 1$st block with the 
number in the $j$th block (that is, concatenate this two 
adjacent blocks) and delete one element from the united block;
finally, decrease by $1$ all elements in the subsequent blocks. 
In this way one obtains an $n$-regular sequence with 
at least $j$ blocks, and it easy to see that every such sequence 
is obtained exactly $e_j$ times.

Thus the expectation to be computed is just
\begin{equation}
\frac{c(n+1,j+1)}{c(n,j)}\,.
\end{equation}
Clearly, $c(n,1) = R(0,n-1,n) = \dbinom{2n-1}{n}$, hence
\begin{equation}
c(n,j)=\binom{2n-1}{n}-\sum _{k=1}^{j-1}\binom{n}{k}\binom{n-1}{k-1}
=\binom{2n-1}{n} + O \left (n^{2j-3} \right),
\end{equation}
as $n \to \infty$. This is asymptotically equal to
\begin{equation}
\frac{\dbinom{2n+1}{n+1}}{\dbinom{2n-1}{n}} = \frac{4n+2}{n+1} = 4 - \frac{2}{n + 1} = 4 + o(1).
\end{equation}
\end{proof}

It is interesting to observe that by \eqref{eq-regrainexp} the average block 
length in a random $n$-regular sequence is 
\begin{equation}
\frac{1}{r}\sum _{j=1}^r e_j = \frac{n}{r(b)}\approx 2
\end{equation}
approximately, as $n \to \infty$. This fact could be interpreted as if
blocks in the beginning of the sequence were significantly
longer. However, fixing $r_n(b) = r$ we find that the lengths of the $r$
blocks are exchangeable random variables with equal expectation $n/r$. At
first sight this two facts seem to be in contradiction. The explanation is that
exchangability only holds conditionally. Blocks in the beginning do
exist even for smaller rainbow numbers, when the average block length is
big, while blocks with large index can only appear when there are many
short blocks in the sequence.

The following assertion gives the number of zerofree sequences and the ratio of the numbers 
of zerofree and regular sequences. 
 
\begin{corollary}\label{corollary-zerofree}
The number of the zerofree $n$-regular sequences $R_z(n)$ is
\begin{equation}
R_z(n) = \binom{2n - 2}{n - 1}\label{eq-zerofreeseq}
\end{equation}
and
\begin{equation}
\lim _{n \rightarrow\infty} \frac{R_z(n)}{R(n)} = \frac{1}{2}. \label{eq-Znlim}
\end{equation} 
\end{corollary}

\begin{proof} \eqref{eq-zerofreeseq} identical with \eqref{eq-lum2}, \eqref{eq-Znlim} is a direct consequence of 
\eqref{eq-lum2} and \eqref{eq-lum3}.
\end{proof}

As the experimental data in Figure \ref{fig-apprrelativ} show, $\frac{E_z(n)}{R(n)} \approx \frac{1}{4}$.

The following lemma allows that the algorithm \textsc{Erd\H os-Gallai-Enumera-} \textsc{ing} 
tests only the zerofree even sequences instead of the even sequences.

\begin{lemma} If $n \geq 2,$ then the number of the $n$-graphical sequences $G(n)$ is  
\begin{equation}
G(n) = G(n - 1) + G_z(n). \label{eq-zerofree}
\end{equation}
\end{lemma}

\begin{proof} If an $n$-graphical sequence $b$ contains at least one zero, that is $b_n = 0,$ then 
$b' = (b_1,\ldots,b_{n-1})$ is $(n - 1)$-graphical or not. If $a = (a_1,\ldots,a_{n-1})$ is an $(n - 1)$-graphical sequence, 
then $a' = (a_1,\ldots,a_{n-1},0)$ is $n$-graphical. 

The set of the $n$-graphical sequences $S$ consists of two subsets: set of zerofree sequences $S_1$ and
the set of sequences $S_2$ containing at least one zero. There is a bijection between the set of the $(n - 1)$-graphical sequences 
and such $n$-graphical sequences, which contain at least one zero. Therefore $|S| = |S_1| + |S_2| = G_z(n) + G(n - 1)$.
\end{proof}

\begin{corollary} If\label{corollary-GGz} $n \geq 1,$ then 
\begin{equation}
G(n) = 1 +  \sum _{i=2}^n G_z(n).\label{eq-GGz}
\end{equation}
\end{corollary}
\begin{proof} Concrete calculation gives $G(1) = 1$. Then using \eqref{eq-zerofree} and induction we get \eqref{eq-GGz}.
\end{proof}

A promising direction of researches connected with the characterization of the function $G(n)$ is the decomposition 
of the even integers into members and the investigation, which decompositions represent a graphical sequence
\cite{BarnesS1995,BarnesS1997,Burns2007}. Using this approach Burns proved the following asymptotic bounds in 2007.

\begin{theorem} \emph{(Burns \cite{Burns2007})} There\label{theorem-Burns} exist such positive constants $c$ and $C$, 
that the following bounds of the function $G(n)$ is true:
\begin{equation}
\frac{4^n}{cn} < G(n) < \frac{4^n}{(\log n)^C \sqrt{n}}.\label{}
\end{equation}
\end{theorem}

\begin{proof} See \cite{Burns2007}.
\end{proof}

This result implies that the asymptotic density of the graphical sequences is zero among the even sequences.

\begin{corollary} If $n \geq 1$, then there exists a positive constant $C$ such that
\begin{equation}
\frac{G(n)}{E(n)} < \frac{1}{(\log _2 n)^C}\label{eq-Burnsass}
\end{equation}
and
\begin{equation}
\lim _{n \to \infty} \frac{G(n)}{E(n)} = 0.\label{eq-Burnsdensity}
\end{equation}
\end{corollary}

\begin{proof} \eqref{eq-Burnsass} is a direct consequence of \eqref{eq-En} and \eqref{eq-Burnsass}, and 
\eqref{eq-Burnsass} implies \eqref{eq-Burnsdensity}.
\end{proof}
As Figure \ref{fig-regevenratio} and Figure \ref{fig-apprrelativ} show, the convergence of the ratio $G(n)/E(n)$ is relative slow. 

\section{Number of graphical sequences \label{sec-gamma}}
\textsc{Erd\H os-Gallai-Enumerating} algorithm (EGE) \cite{IvanyiLMS2011AML} generates and tests for given $n$ 
every zerofree even sequence. Its input is $n$ and output is the number of corresponding zerofree graphical sequences $G_z(n).$

The algorithm is based on \textsc{Erd\H os-Gallai-Linear} algorithm. It generates and tests only the zerofree even 
sequences, that is according to Corollary \ref{cor-zerofree} and Figure \ref{fig-apprrelativ} about the 25 percent of the 
$n$-regular sequences. 

EGE tests the input sequences only in the checking points. Corollary \ref{cor-regrain} shows 
that about the half of the elements of the input sequences are check points.  

Figure \ref{fig-apprrelativ} contains data showing that EGE investigates even less than the half 
of the elements of the input sequences.

Important property of EGE is that it solves in $O(1)$ expected time 
\begin{enumerate}
\addtolength{\itemsep}{-0.6\baselineskip}
\item
the generation of one input sequence;
\item
the updating of the vector $H$;
\item
the updating of the vector of checking points;
\item
the updating of the vector of the weight points.
\end{enumerate}

We implemented the parallel version of EGE (EGEP). It was run on about 200 PC's containing about 700 cores. The total 
running time of EGEP is contained in Figure \ref{fig-totalpar}

\begin{figure}[!t]
\begin{small}
\begin{center}
\begin{tabular}{||c|c|c||}  \hline \hline
$n$  & running time (in days)       & number of slices   \\ \hline \hline
$24$ &              $7$              & $415$           \\ \hline
$25$ &              $26$              & $415$           \\ \hline
$26$ &             $70$             & $435$ \\ \hline
 $27$ &            $316$            & $435$ \\ \hline
 $28$ &               $1130$             & $2 001$ \\ \hline
 $29$ &               $6733$             & $15 119$\\ \hline \hline
\end{tabular}
\end{center} 
\end{small}
\caption{The runnng time of EGEP for $n = 24, \ \ldots, \ 29$\label{fig-totalpar}}
\end{figure}

The pseudocode of the algorithm see in \cite{IvanyiLMS2011AML}.
The amortized running time of this algorithm for a sequence is $\Theta(1)$, so the total running time of the whole program 
is $O(E(n))$. 

\section{Summary\label{sec-summary}}

In Figure \ref{fig-regevenratio} the values of $R(n)$ up to $n = 24$ are the elements of the sequence A001700 of 
OEIS \cite{Sloane2011A001700}, the values of $E(n)$ up to $n = 23$ are the elements of the sequence A005654 
\cite{Sloane2011A005654} of the OEIS, and in Figure \ref{fig-numberofseq} the values $G(n)$ 
are up to $n = 23$ are the elements of sequence  
A0004251-es \cite{Sloane2011A004251} of OEIS. The remaining values are new \cite{IvanyiLMS2011AML,IvanyiL2011CEJOR}.

Figure \ref{fig-numberofseq} contains the number of graphical sequences $G(n)$ for $n = 1, \ \ldots, \ 29,$ 
and also $G(n + 1)/G(n)$ for $n = 1, \ \ldots, \ 28.$ 

The referenced manuscripts, programs and further simulation results can be found at the homepage of the authors, 
among others at  
\url{http://compalg.inf.elte.hu/~tony/Kutatas/EGHH/}

\section*{Acknowledgements} The authors thank Zolt\'an Kir\'aly (E\"otv\"os Lor\'and University, Faculty of Science, 
Dept. of Computer Science) for his advice concerning the weight points, Antal S\'andor and his colleagues 
(E\"otv\"os Lor\'and University, Faculty of Informatics), further \'Ad\'am M\'anyoki 
(TFM World Kereskedelmi \'es Szolg\'altat\'o Kft.) for their help in running of our 
timeconsuming programs and the unknown referee for the useful corrections. 
The European Union and the European Social Fund have provided financial support to the project 
under the grant agreement no. T\'AMOP 4.2.1/B-09/1/KMR-2010-0003.

\bigskip
\rightline{\emph{Received: September 30, 2011 {\tiny \raisebox{2pt}{$\bullet$\!}} Revised: November 10, 2011}} 

\end{document}